\documentclass[11pt]{article}
\usepackage{graphicx,amssymb,amsmath}
\usepackage[margin=1.1in,a4paper]{geometry}
\usepackage[croatian,german,english]{babel}
\usepackage{amssymb}
\usepackage{url}
\newcommand{\realrange}[2]{\left[#1, #2\right]}

\newcommand{\unitrange}[2]{\realrange{0}{1}}

\newcommand{\llabel}[1]{\label{\labelprefix:#1}}
\newcommand{\labelprefix}{} 

\newcommand{\discussionsize}{\small}

\newenvironment{code}{\noindent
\begin{tabbing}%
\hspace{1em}\=\hspace{1em}\=\hspace{1em}\=\hspace{1em}\=\hspace{1em}\=%
\hspace{1em}\=\hspace{1em}\=\hspace{1em}\=\hspace{1em}\=\hspace{1em}\=%
\kill}{\end{tabbing}}

\newcommand{\labelcommand}{}
\newcommand{\captiontext}{}
\newsavebox{\codeparam}
\newcounter{lineNumber}
\newenvironment{disscodepos}[3]{%
\renewcommand{\labelcommand}{#2}%
\renewcommand{\captiontext}{#3}%
\sbox{\codeparam}{\parbox{\textwidth}{#3}}%
\begin{figure}[#1]\begin{center}\begin{code}\setcounter{lineNumber}{1}}{%
\end{code}\end{center}\caption{\llabel{\labelcommand}\captiontext}\end{figure}}

{\end{disscodepos}}

\newdimen\endofsize\endofsize=0.5em
\def\endofbeweis{~\quad\hglue\hsize minus\hsize
                 \hbox{\vrule height \endofsize width
\endofsize}\par}

\def\see{{\thicksim}}

\newcommand{\ignore}[1]{}

\title{Improved Approximations for Guarding \\ 1.5-Dimensional Terrains }

\author{Khaled Elbassioni\thanks{Max-Planck-Institut f\"ur
	Informatik, Saarbr\"ucken, Germany.}
        \and
        Domagoj Matijevi\'c\thanks{Department of Mathematics,
		J. J. Strossmayer University of Osijek, Croatia.}
		\and
		Juli\'{a}n Mestre\footnotemark[1]\ \thanks{Research supported by an
                  Alexander von Humboldt Fellowship.}
    \and
		Domagoj \v{S}everdija\footnotemark[2]
}

\date{}

\index{Elbassioni, Khaled}
\index{Matijevi\'c, Domagoj}
\index{Mestre, Juli\'{a}n}
\index{\v{S}everdija, Domagoj}

\newcommand{\sees}{\thicksim}

\newcommand{\qed}{\hfill$\square$\bigskip}

\newcommand{\hide}[1]{}

\newtheorem{theorem}{Theorem}

\newtheorem{lemma}{Lemma}

\newtheorem{proposition}{Proposition}

\begin{document}

\maketitle

\begin{abstract}
  We present a 4-approximation algorithm for the problem of placing the fewest
  guards on a 1.5D terrain so that every point of the terrain is seen by at
  least one guard.  This improves on the currently best approximation factor
  of 5 due to {\it King (LATIN 2006, pages 629-640)}. Unlike previous
  techniques, our method is based on rounding the linear programming
  relaxation of the corresponding covering problem. Besides the simplicity of
  the analysis, which mainly relies on decomposing the constraint matrix of
  the LP into totally balanced matrices, our algorithm, unlike previous work,
  generalizes to the weighted and partial versions of the basic problem.
\end{abstract}

\section{Introduction}
\label{sec:introduction}

In the \emph{1.5D terrain guarding} problem we are given a polygonal region in
the plane determined by an $x$-monotone polygonal chain, and the objective is
to find the minimum number of guards to place on the chain such that every
point in the polygonal region is guarded. This kind of guarding problems and
its generalizations to 3-dimensions are motivated by optimal placement of
antennas for communication networks; for more details see
\cite{conf/cccg/ChenEU95,journals/siamcomp/Ben-MosheKM07} and the references
therein.

One can easily see that one point is enough to guard the polygonal region if
we are allowed to select guards anywhere in the plane. However, the problem
becomes interesting if guards can only be placed on the boundary chain. Under
this restriction, two natural versions of the problem arise: in the
\emph{continuous} version the guards can be placed anywhere along the chain
and all points in the terrain must be guarded, while in the \emph{discrete}
version the guards and points to be guarded are arbitrary subsets of the
chain.

\subsection{Previous Work}
\label{sec:prev-work}

Chen et al.~\cite{conf/cccg/ChenEU95} claimed that the 1.5D-terrain guarding
problem is NP-hard, but a complete proof of the claim was never published
\cite{conf/cccg/DemaineO06,conf/latin/King06,journals/siamcomp/Ben-MosheKM07}. They
also gave a linear time algorithm for the \emph{left-guarding} problem, that
is, the problem of placing the minimum number of guards on the chain such that
each point of the chain is guarded from its left.  Based on purely geometric
arguments, Ben-Moshe et al.~\cite{journals/siamcomp/Ben-MosheKM07} gave the
first constant-factor approximation algorithm for the 1.5D-terrain guarding
problem. Although they did not state the value of the approximation ratio
explicitly, it was claimed to be at least 6
in~\cite{conf/latin/King06}. Clarkson et al.~\cite{conf/compgeom/ClarksonV05}
gave constant factor approximation algorithms for a more general class of
problems using $\epsilon$-nets and showed that their technique can be used to
get a constant approximation for the 1.5D-terrain guarding problem. Most
recently, King~\cite{conf/latin/King06} claimed that the problem can be
approximated with a factor of 4, but the analysis turned out to have an error
that increases the approximation factor to $5$ \cite{misc/errata/King}.

\subsection{Our results and outline of the paper}
\label{sec:our-results}

The main building block of our algorithms is an LP-rounding algorithm for
one-sided guarding: A version of the problem where a guard can see either to
the left or to the right. Guided by an optimal fractional solution, we can
partition the points into those that should be guarded from the left, and
those that should be guarded from right. This turns out to be a very useful
information since we can show that the LPs for the left-guarding and
right-guarding problems are integral. We prove this by establishing a
connection between the guarding problem and totally balanced covering problems
that is of independent interest. Altogether, this leads to a factor 2
approximation for one-sided guarding. Then we show how to reduce other
variants of the problem to the one-sided case by incurring an extra factor of
2 in the approximation ratio.

A nice feature of this framework is that the algorithms emanating from it, as
well as their analysis, are very simple. This comes in contrast with the
relatively complicated algorithms of
\cite{journals/siamcomp/Ben-MosheKM07,conf/latin/King06} whose
description/analysis involves a fairly long list of cases.  In addition, our
framework allows us to tackle more general versions of the problem than those
considered in the literature thus far; for example, guards can have weights
and we want to minimize the weight of the chosen guards, or where we are not
required to cover all the terrain, but only a prescribed fraction of it. It
seems that such variants are very difficult to deal with, if one tries to use
only geometric techniques as the ones used in
\cite{journals/siamcomp/Ben-MosheKM07,conf/latin/King06} for the basic
problem.

It is worth noting that the idea of using the fractional solution to the
LP-covering problem to partition the problem into several integral subproblems
has been used before
\cite{journals/orl/HassinS08,conf/stacs/Mestre08,journals/jal/GaurIK02}.

In the next section, we
define the basic guarding problem and its variants
more formally. In Section~\ref{sec:left-guarding} we focus on the left
guarding problem and show that this is a totally balanced covering problem.
Section~\ref{sec:one-sided-guarding} shows how to get a 2-approximation for
one-sided guarding. Finally, in Section~\ref{sec:applications} we apply these
results to obtain constant-factor approximation algorithms for more general variants of
the guarding problem.

\section{Preliminaries}
\label{sec:preliminaries}

A \emph{terrain} $T$ is an $x$-monotone polygonal chain with $n$ vertices,
i.e., a piecewise linear curve intersecting any vertical line is in at most
one point. Denote by $V$ the vertices of $T$ and by $n=|V|$ the complexity of
the chain. The terrain polygon $P_T$ determined by $T$ is the closed region in
the plane bounded from below by $T$.

For two points $p$ and $q$ in $P_T$, we say that $p$ \emph{sees} $q$ and
write $p \sees q$, if the line segment connecting $p$ and $q$ is contained
in $P_T$, or equivalently, if it never goes strictly below $T$.

The $1.5D$-terrain guarding problem for $T$ is to place guards on $T$ such
that every point $p\in P_T$ is seen by some guard. One can easily see, by the
monotonicity of $T$, that any set of guards that guards $T$ is also enough to
guard $P_T$. Henceforth we restrict our attention to the case when the
requirement is to guard all points of $T$.

The \emph{continuous} $1.5D$-terrain guarding problem is to select a smallest
set of guards $A \subseteq T$ that sees every point in $T$; in other words,
for every $p \in T$ there exists $g \in A$ such that $g \sees p$. We also
consider the following variants of this basic problem:

\begin{enumerate}
  \item In the \emph{discrete} version we are given a set of possible
  guards $G \subseteq T$ with weights $w: G \rightarrow R^+$ and a set of
  points $N \subseteq T$. The goal is to select a minimum weight set of
  guards $A \subseteq G$ to guard $N$.

  \item In the \emph{partial} version we are given a profit function $p: N
  \rightarrow R^+$ and a budget $b$. The goal is to find a minimum weight set of
  guards such that the profit of unguarded points is at most $b$. In the
  continuous variant, $b$ is the length of $T$ that can be left unguarded.

  \item In the \emph{one-sided guarding} version the guards can see in only one
  of two directions: left or right. Specifically, given $3$ sets of points $N$, $G_L$ and $G_R$,
  we want to find sets $A_L\subseteq G_L$ and $A_R\subseteq G_R$ of guards
  such that for all $p \in N$ there is $g \in A_L$ such that $g < p$ and $g
  \sees p$, or $g \in A_R$ such that $g > p$ and $g \sees p$. The sets $G_L$
  and $G_R$, and hence $A_L$ and $A_R$ need not be disjoint. The overall cost of the solution
  is $w(A_L) + w(A_R)$.

  This includes both the \emph{left-} and {\em right-}guarding versions
  where guards in the given set $G$ can see only from the left, respectively, right
  (setting $G_L=G$ and $G_R=\emptyset$ we get the left-guarding problem, while
  setting $G_R=G$ and $G_L=\emptyset$ gives the right-guarding problem).
\end{enumerate}

Using a unified framework we get 4-approximations for nearly all\footnote{The
  only exception is instances of the discrete variant when $G \cap N \neq
  \emptyset$. Here we get a 5-approximation.} of these variants.  Our approach
is based on linear programming, totally balanced matrices, and the paradigm of
rounding to an integral problem
\cite{journals/jal/GaurIK02,journals/orl/HassinS08}. We progressively build
our approximations by reducing each variant to a simpler problem. First, we
start establishing a connection between the left-guarding problem and totally
balanced matrices. Then, we show how to use this to get a 2 approximation for
the one-sided guarding. Finally, we show how the latter implies a 4
approximations for other variants.

Throughout the paper we will make frequent use of the following easy-to-prove
claim.

\begin{lemma}[\cite{journals/siamcomp/Ben-MosheKM07}] \label{lem:Mitchell}
  Let $a<b<c<d$ be four points on $T$. If $a\sees c$ and $b\sees d$,
  then $a\sees d$.
\end{lemma}

Let $S(p)=\{g\in G \mid g \sees p \}$ be the set of guards that
see point $p\in N$.  Denote by $S_L(p)=\{q\in G \mid g < p \textrm{
and } g\sees p \}$ the set of guards that see $p$ strictly from
the left, and analogously by $S_R(p)$ the set of guards that see $p$
strictly from the right.

\section{Left-guarding and totally balanced matrices}
\label{sec:left-guarding}

Even though this section deals exclusively with the left-guarding version, it
should be noted that everything said applies, by symmetry, to the
right-guarding version. Recall in this case that we are given two sets of points $N,G,$ where each point in $N$
has to be guarded using only guards from $G$ that lie strictly to its left.

Consider the following linear programming formulation.

\begin{gather} \label{LP:left} \tag{LP1}
  \text{minimize} \ \sum_{g\in G} w_g\, x_g \\[-6ex] \notag
\end{gather}
\hspace{1.5cm} subject to \\[-3ex]
\begin{align}
  \label{LP-left:coverage-constraint}
  \sum_{g\in S_L(p)} x_{g}\ & \ge 1 & \forall p\in N \\
  \notag
  x_g\ & \geq 0 & \forall g\in G
\end{align}

Variable $x_g$ indicates whether $g$ is chosen as a guard. Constraint
\eqref{LP-left:coverage-constraint} asks that every point is seen by some
guard from the left.

Let $A\in\{0,1\}^{n \times m}$ be a binary matrix. Call $A$ a
\emph{left-visibility} matrix if it corresponds to the element-set incidence
matrix of the coverage problem defined by \eqref{LP:left} for some instance of
the left-guarding problem. Also, $A$ is said to be \emph{totally
  balanced}~\cite{journal/mp/Berge72} if it does not contain a square
submatrix with rows and columns sums equal to 2 and no identical columns.
Finally, $A$ is in standard greedy form if it does not contain as an induced
submatrix
\begin{equation} \label{eq:forbidden}
  \left[\begin{array}{cc} 1 & 1 \\ 1 & 0 \end{array} \right].
\end{equation}

An equivalent characterization~\cite{journal/SIADM/HoffmanKS85} is that $A$ is
totally balanced if and only if $A$ can be put into greedy standard form by
permuting its rows and columns.

\begin{lemma} \label{lem:split}
  Any left-visibility matrix is totally balanced.
\end{lemma}

\begin{proof}
  Let $A$ be a left-visibility matrix. We show how to put $A$ into standard
  greedy form. Permute the rows and columns of $A$ such that all points are
  ordered from left to right and all guards are ordered from right to
  left. Suppose that there exists an induced $2\times 2$ sub-matrix of the
  form \eqref{eq:forbidden}, whose rows are indexed by $p_1,p_2\in N$, and
  whose columns are indexed by $g_1,g_2\in G$. Then we have the following
  order: $g_2<g_1<p_1<p_2$. Now we apply Lemma~\ref{lem:Mitchell}
  with $a=g_2$, $b=g_1$, $c=p_1$ and $d=p_2$ to arrive at the
  contradiction $p_2\sees g_2$. \qed
\end{proof}

It is known that for a totally balanced matrix $A$, the
polyhedron $\{x\geq 0~:~Ax\geq 1\}$ is integral. Furthermore, there is an
efficient purely combinatorial algorithm for finding an optimal integral
solution to \eqref{LP:left} due to Kolen \cite{thesis/Kolen82}. Indeed, in the
next subsection we show that this algorithm translates into an extremely
simple procedure for the uniform weight case, i.e., when $w_g = 1$ for all $g
\in G$.

\subsection{Uniform left-guarding}

For each point $p\in N$ let $L(p)$ denote the left-most guard that
sees $p$. Consider the simple greedy algorithm on the set of points $N$ shown
below: Points in $N$ are scanned from left to right
and when we find an unguarded point $p$, we select $L(p)$ as a guard.

\begin{center}
  \begin{minipage}{5em}
    \begin{code}
      {\sc left-guarding} $(T,N,G)$ \\
      $A \leftarrow \emptyset$ \\
      {\bf for} $p\in N$ processed from left to right \\
      \>  {\bf if} $p$ is not yet seen by $A$ {\bf then} \\
      \> \> $A \leftarrow A \cup \{L(p)\}$ \\
      {\bf return} A
    \end{code}
  \end{minipage}
\end{center}

The algorithm can be implemented in $O(|N|\log|G|)$ time using a procedure
similar to Graham's scan~\cite{journals/ipl/Graham72} for convex-hull
computation. To see the it returns an optimal solution, let $X \subseteq N$ be
those points that force the algorithm to add a guard. Suppose, for the sake of
contradiction, that there exist two points $p'$ and $p''$ in $X$ that are seen
from the left by the same guard $g \in G$, in other words, $g< p'<p''$ and $g
\sees p'$ and $g \sees p''$. Let $g'=L(p')$, and note that $g'
\le g$. If $g' = g$ then $g'\sees p''$ and therefore $p''$ would have not been
unguarded when it was processed. Hence $g'<g$, but Lemma~\ref{lem:Mitchell}
tells us that $g'\sees p''$ and we get a contradiction. Therefore, each guard
in $G$ can see at most one point in $X$, which means $|X|$ is a lower bound on
the optimal solution. Since the cardinality of $A$ equals that of $X$, it
follows that $A$ is optimum.

\section{A $2$-approximation for one-sided guarding}
\label{sec:one-sided-guarding}

In this section we study discrete weighted one-sided guarding. Recall that in
this variant, we are given a set of points $N$ and two sets of guards $G_L$
and $G_R$, where each guard in $G_L$ (respectively, $G_R$) can only guard
points from $N$ strictly to its right (respectively, strictly to its left).
We assume without loss of generality that, each point in $N$ can be either
seen by a guard on its left or by a guard on is right, for otherwise, it must
be guarded by itself and the system is infeasible, a situation which can be
discovered in a preprocessing step.

We state our main result and then describe the algorithm.

\begin{theorem} \label{thm:one-sided}
  There is a 2-approximation algorithm for discrete one-sided guarding.
\end{theorem}

Consider the following LP for finding the optimal set of left and right guards:

\begin{gather} \label{LP:one-sided} \tag{LP2}
  \text{minimize} \ \sum_{g\in G_L} w_g\, x_{g,L} +  \sum_{g\in G_R}w_g\, x_{g,R} \\[-6ex] \notag
\end{gather}
\hspace{1.5cm} subject to \\[-3ex]
 \begin{align}
   \label{LP2:coverage-constraint}
\sum_{g\in S_L(p)\cap G_L} x_{g,L}\ + \sum_{g \in S_R(p)\cap G_R} x_{g,R}\  & \ge 1 & \forall p\in N \\
   \notag
   x_{g,L}\ & \ge 0 & \forall g\in G_L \\
   \notag
   x_{g,R} \ & \ge 0 & \forall g \in G_R
\end{align}

Variable $x_{g,L}$ indicates whether $g$ is chosen in $A_L$ and $x_{g,R}$
indicates whether $g$ is chosen in $A_R$.  Constraint
\eqref{LP2:coverage-constraint} asks that every point is seen by some guard,
either from the left or from the right.

The algorithm first finds an optimal fractional solution $x^*$
to~\eqref{LP:one-sided}. Guided by $x^*$, we divide the points into two
sets
\begin{align*}
  N_L &= \left \{p\in N \mid \textstyle \sum_{g\in S_L(p)\cap G_L} x_{g,L}^* \ge
    \frac12 \right\} \text{, and} \\
  N_R &= \left\{p\in N \mid \textstyle \sum_{g\in S_R(p)\cap G_R}x^*_{g,R} \ge
    \frac12 \right\}.
\end{align*}

Using the results from Section~\ref{sec:left-guarding}, we solve optimally the
left-guarding problem for the pair $(N_L,G_L)$ and the right-guarding problem
for the pair $(N_R,G_R)$. This gives us two sets of guards $A_L^*$ and $A_R^*$.
The final solution is a combination of these two.

It is easy to construct examples where solving separately the left-guarding
and right-guarding problems and then taking the minimum of these two solutions
is arbitrarily far from the optimal value. The intuition behind the algorithm
is to use the LP solution to determine which points should be guarded from the
left and which should be guarded from the right. The fractional solution also
allows us to bound the cost of $A_L^*$ and $A_R^*$.

\begin{lemma} \label{lem:AL}
  Let $A_L^*$ and $A_R^*$ be optimal solutions for the pairs $(N_L,G_L)$ and $(N_R,G_R)$
  respectively. Then $w(A_L^*) \leq 2 \sum_{g \in
    G_L} w_g\, x^*_g$ and $w(A_R^*) \leq 2 \sum_{g \in G_R} w_g \, x^*_g$.
\end{lemma}

\begin{proof}
  We only prove the first inequality as the second is symmetrical.
  Setting $x_{g,L} = 2 x^*_g$ we get a fractional solution for \eqref{LP:left}
  for guarding $N_L$. The solution $x$ is feasible, by definition of $N_L$,
  and its cost is $2 \sum_{g \in G_L} w_g\, x^*_g$. Therefore, the optimal
  fractional solution can only be smaller than that. Lemma~\ref{lem:split} tells
  us that the cost of an optimal fractional solution is the same as the cost
  of an optimal integral solution, namely, $w(A^*_L)$. \qed
\end{proof}

Since $\sum_{g \in G_L} w_g\, x^*_{g,L} + \sum_{g \in G_R} w_g \, x^*_{g,R}$
is a lower bound on the cost of the optimal solution for guarding $N$, it follows that
the cost of $(A_L^* , A_R^*)$ is at most twice the optimum. To see that this
is feasible, consider some point $p \in N$. Because of
\eqref{LP2:coverage-constraint} and our assumption that each point is seen by
some guard on its left or on its right, it must be the case that $p \in N_L$
or $p \in N_R$. Therefore $p$ must be covered, either from the left by $A_L^*$
or from the right by $A_R^*$.

To compute $A^*_L$ and $A^*_R$ we can take the fractional solution to
\eqref{LP:left} and turn it into a basic, and therefore integral, solution
without increasing its cost. Alternatively, we can run Kolen's
algorithm~\cite{thesis/Kolen82} for matrices in greedy standard form. This
finishes the proof of Theorem~\ref{thm:one-sided}.

\subsection{Partial covering}

In this section we focus on the partial version of the one-sided guarding problem.

\begin{theorem} \label{thm:partial}
  There is a $(2 + \epsilon)$-approximation and a quasi-polynomial time
  $2$-approximation for partial discrete one-sided guarding.
\end{theorem}

Our approach is based on the framework of Mestre
\cite{conf/stacs/Mestre08}. We say $A$ is a \emph{one-sided-visibility} matrix
if it is the element-set incidence matrix of the covering problem defined by
\eqref{LP:one-sided} for some instance of the one-sided guarding problem.
Also, $A$ is said to be \emph{2-separable} if there exist binary matrices $A_1$
and $A_2$ such that $A = A_1 + A_2$ and every matrix $B$ formed by taking rows
from $A_1$ or $A_2$ is totally balanced (the $i$th row of $B$ is the $i$th row
of $A_1$ or the $i$th row of $A_2$, for all $i$).

\begin{proposition}[\cite{conf/stacs/Mestre08}] \label{prop:partial}
  Let $A$ be a 2-separable matrix. Then there is a $(2 + \epsilon)$-approximation
  and a quasi-polynomial time $2$-approximation for the partial
  problem defined by $A$.
\end{proposition}

Therefore, all we need to do to prove Theorem~\ref{thm:partial} is to argue
that every one-sided visibility matrix is 2-separable.

\begin{lemma}
  Any one-sided visibility matrix is 2-separable.
\end{lemma}

\begin{proof}
  Let $A$ be a one-sided visibility matrix and assume, without loss of
  generality, that $A$ has the form $[C_1\ C_2]$ where the columns of $C_1$
  correspond to left guards $G_L$ and the columns of $C_2$ correspond to the right
  guards $G_R$.

  Our decomposition of $A$ uses $A_1 = [C_1\ 0]$ and $A_2 = [0\ C_2]$. Suppose
  that a matrix $B$ is formed by taking rows from $A_1$ and $A_2$. Let $N_L$
  be the set of rows originating from $A_1$ and $N_R$ the set of rows
  originating from $N_R$ (note that $N_L$ and $N_R$ constitute a partition
  of $N$). Permute the rows of $B$ so that rows in $N_L$ appear before rows
  in $N_R$. This gives rise to the following block matrix
  \[ B' = \left[ \begin{array}{cc} D_1 & 0 \\ 0 & D_2 \end{array} \right] \]
  where the rows of $D_1$ correspond to points in $N_L$ and its columns to
  left guards, and the rows of $D_2$ correspond to points in $N_R$ and its
  columns to right guards. By Lemma~\ref{lem:split} both $D_1$ and $D_2$ are
  totally balanced. It follows that $B'$ must be totally balanced as well. \qed
\end{proof}

This finishes the proof of Theorem~\ref{thm:partial}.

\section{Applications}
\label{sec:applications}

In this section we show how to use the 2-approximations for one-sided guarding
to design good approximation algorithms for more general variants.

\subsection{The continuous case}

We assume that the weights are uniform\footnote{This assumption can be removed
  using standard discretization techniques at the expense of a small increase
  in the approximation factor}.  Recall that in this variant guards can be
placed anywhere on the terrain and we are to guard all the points. Our
reduction to the discrete case follows the approach of Ben-Moshe et
al.~\cite{journals/siamcomp/Ben-MosheKM07}.

\begin{theorem} \label{thm:continuous}
  There is a $4$-approximation algorithm for the continuous case and and a $(4 +
  \epsilon)$-approximation for its partial version.
\end{theorem}

Let $A^*$ be an optimum set of guards for a given instance $T$ of the
continuous problem. Consider a guard $g$ in $A^*$. If $g$ is not a vertex
of $T$ then it must lie on a segment $\overline{pq}$ of $T$. Suppose without loss of generality that $p<q$,
then a left guard at $p$ and a right guard at $q$
can see at least as much as $g$ does. If $g$ is a vertex of $T$ then a left
guard and a right guard at $g$ together can see the same as $g$ does minus $g$
itself. Therefore there exists a solution $A'$ that uses only left and right
guards on the vertices of $T$ that covers $T \setminus V$ such
that $|A'| = 2 |A^*|$.

\begin{figure}[t]
  \centering
  \includegraphics[width=.6\columnwidth]{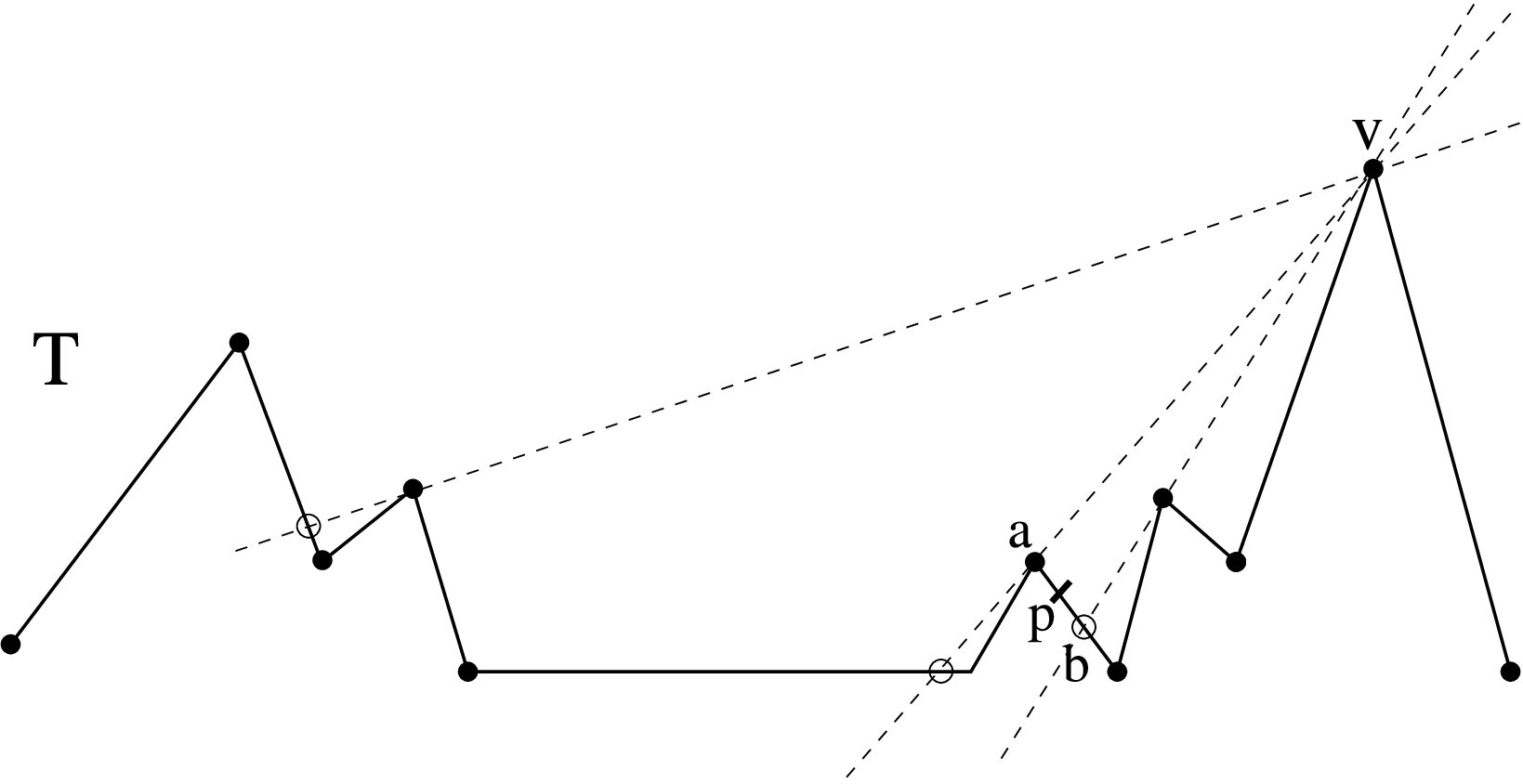}
  \caption{Example of additional set of points/guards for a vertex $v$ of $T$.
  Point $p$ is the point selected from the essential segment $\overline{ab}$.}
  \label{fig:additionalPoints}
\end{figure}

To deal with the fact that every point must be guarded, consider the line
through each pair of vertices $v_1,v_2 \in V$ such that $v_1\see v_2$ and
introduce at most two new points that see $v_1$ and $v_2$ at the place where
the line intersects the terrain. These points partition $T$ into $O(n^2)$
\emph{essential segments}. In the strict interior of each segment introduce an
additional point $p$ that is \emph{responsible} for the segment. Let $M$ be
the set of all such points. (See Figure~\ref{fig:additionalPoints} for an illustration.) The key realization is that for every guard $g \in V$ and essential
segment $(a,b)$, either $g$ can see the whole segment or nothing of it.

Hence, a feasible solution to the one-sided discrete version with $G_L=G_R = V$
and $N= M$ also constitutes a feasible solution to the continuous case. Let $A''$ be an optimal solution for this discrete problem, and $A'''$ be
the solution returned by Theorem~\ref{thm:one-sided}. Since $A'$ is feasible for the discrete instance, we get
$|A'''|\leq 2|A''|\leq 2|A'|= 4|A^*|$ and we get
an overall approximation factor of 4.

For the partial version where we want at most a fraction of the length to be
left unguarded we give to each point in $p \in M$ a profit equal to the length
of the essential segment it is responsible for.

\subsection{The discrete case}

We consider the discrete version where we are given a set of guards $G$ and
set of points $N$ to guard. In this case, guards can see in both directions.

\begin{theorem} \label{thm:discrete}
  There is a $4$-approximation for the weighted discrete case and $(4 +
  \epsilon)$-approximation for its partial version when $G \cap N =
  \emptyset$. Otherwise, we get $5$ and $(5 + \epsilon)$-approximations respectively.
\end{theorem}

The case where $G \cap N = \emptyset$ is easily handled by replacing a guard
who can see in both directions with two guards that can see in one
direction. Thus we pay a factor 2 to reduce the general problem to one-sided
guarding. This also holds for the partial version.

Notice that if $G \cap N \neq \emptyset$ then the reduction above must pay a
factor of 3 since a point guarding itself must be guarded by some other point
strictly from the left or the right, and thus it only leads to a 6
approximation. To get the ratio of 5 we need to use yet another linear
program.

\begin{gather} \label{LP:discrete} \tag{LP3}
  \text{minimize} \ \sum_{g\in G} w_g\, x_g \\[-6ex] \notag
\end{gather}
\hspace{1.5cm} subject to \\[-3ex]
\begin{align}
  \sum_{g\in S(p)} x_{g}\ & \ge 1 & \forall p\in N \notag \\
  \notag
  x_g\ & \geq 0 & \forall g\in G
\end{align}

Let $x^*$ be an optimal fractional solution to \eqref{LP:discrete}. As in the
one-sided case we will let the solution $x$ dictate which points should be
self-guarded and which should be guarded by others. Define
\begin{equation*}
  A_0 = \left \{p\in N \cap G \mid \textstyle x_{p}^* \ge
    \frac15 \right\}.
\end{equation*}

We place guards at $A_0$ at a cost of most
\begin{equation}
   \label{eq:self-cost}
   w(A_0) \leq 5 \sum_{g \in A_0} w_g x^*_g.
\end{equation}

Let $N'$ be the set of points in $N$ not seen by $A_0$ and let $G' = G
\setminus A_0$. We will construct a fractional solution for the one-sided
guarding problem for $N'$ and $G_L=G_R=G'$. For each $g \in G'$
let $x_{g,L} = x_{g,R} = \frac{5}{4} x^*_g$. The fractional solution $x$ is
feasible for \eqref{LP:one-sided} since
\begin{equation*}
  \sum_{g\in S_L(p)\cap G_L} x_{g,L}\ + \sum_{g \in S_R(p)\cap G_R} x_{g,R} = \frac{5}{4} \sum_{g \in S(p)
  \setminus \{ p \}} x^*_g \geq  \frac54 \left(1 - \frac15\right) = 1,
\end{equation*}
for all $p\in N'$.
Let $(A^*_L,A^*_R)$ be the solution found for the one-sided problem. The
cost of these sets of guards is guaranteed to be at most twice that of $x$,
which in turn is $\frac52 \sum_{g \in G \setminus
  A_0} w_g x_g^*$. Thus the overall cost is
\begin{equation}
   \label{eq:one-sided-cost}
   w(A^*_L) + w(A^*_R) \leq 5 \sum_{g \in G \setminus A_0} w_g x^*_g.
\end{equation}

Hence, the second part of Theorem~\ref{thm:discrete} follows from
\eqref{eq:self-cost} and~\eqref{eq:one-sided-cost}. Finally, for the partial
version we note the proof of Proposition~\ref{prop:partial} uses as a lower
bound the cost of the optimal fractional solution, so cost of the solution
returned can still be related to the cost of $x$, which is necessary to get
the stated approximation guarantee.

\section{Conclusion}

We gave a 4-approximation for the continuous 1.5D terrain guarding problem as
well as several variations of the basic problem. Our results rely, either
explicitly or implicitly, on the LP formulation~\eqref{LP:discrete} for the
discrete case. Very recently, King~\cite{conf/cccg/King08} showed that the VC
dimension of the discrete case is exactly 4. More precisely, he showed a
terrain with 4 guards and 16 points (these sets are disjoint) such that each
point is seen by a different subset of the guards. If we have to cover the
points that are seen by pairs of guards, we get precisely a vertex cover
problem on the complete graph with 4 vertices. An integral solution must pick
3 vertices, while a fractional solution can pick a half of all vertices. It
follows that
the integrality gap of~\eqref{LP:discrete} is at least $3/2$, even when $G
\cap N =
\emptyset$. On the other hand, our analysis shows that the gap is at most 4. We leave as an open problem to determine the exact integrality gap of \eqref{LP:discrete}.


\bibliographystyle{abbrv}
\bibliography{paper}

\begin{thebibliography}{10}

\bibitem{journals/siamcomp/Ben-MosheKM07}
B.~Ben-Moshe, M.~J. Katz, and J.~S.~B. Mitchell.
\newblock A constant-factor approximation algorithm for optimal 1.5d terrain
  guarding.
\newblock {\em {SIAM} Journal on Computing}, 36(6):1631--1647, 2007.

\bibitem{journal/mp/Berge72}
C.~Berge.
\newblock Balanced matrices.
\newblock {\em Mathematical Programming}, 2:19--31, 1972.

\bibitem{conf/cccg/ChenEU95}
D.~Z. Chen, V.~Estivill-Castro, and J.~Urrutia.
\newblock Optimal guarding of polygons and monotone chains.
\newblock In {\em Proceedings of the 7th Canadian Conference on Computational
  Geometry}, pages 133--138, 1995.

\bibitem{conf/compgeom/ClarksonV05}
K.~L. Clarkson and K.~R. Varadarajan.
\newblock Improved approximation algorithms for geometric set cover.
\newblock In {\em Proceedings of the 20th Symposium on Computational Geometry},
  pages 135--141, 2005.

\bibitem{conf/cccg/DemaineO06}
E.~D. Demaine and J.~O'Rourke.
\newblock Open problems: Open problems from cccg 2005.
\newblock In {\em Proceedings of the 18th Canadian Conference on Computational
  Geometry}, pages 75--80, 2006.

\bibitem{journals/jal/GaurIK02}
D.~R. Gaur, T.~Ibaraki, and R.~Krishnamurti.
\newblock Constant ratio approximation algorithms for the rectangle stabbing
  problem and the rectilinear partitioning problem.
\newblock {\em Journal of Algorithms}, 43(1):138--152, 2002.

\bibitem{journals/ipl/Graham72}
R.~L. Graham.
\newblock An efficient algorithm for determining the convex hull of a finite
  planar set.
\newblock {\em Information Processing Letters}, 1(4):132--133, 1972.

\bibitem{journals/orl/HassinS08}
R.~Hassin and D.~Segev.
\newblock Rounding to an integral program.
\newblock {\em Operations Research Letters}, 36(3):321--326, 2008.

\bibitem{journal/SIADM/HoffmanKS85}
A.~J. Hoffman, A.~Kolen, and M.~Sakarovitch.
\newblock Totally-balanced and greedy matrices.
\newblock {\em {SIAM} Journal on Algebraic and Discrete Methods}, 6:721--730,
  1985.

\bibitem{misc/errata/King}
J.~King.
\newblock {Errata on ``A 4-Approximation Algorithm for Guarding 1.5-Dimensional
  Terrains''}.
\newblock \url{http://www.cs.mcgill.ca/~jking/papers/4apx_latin.pdf}.

\bibitem{conf/latin/King06}
J.~King.
\newblock A 4-approximation algorithm for guarding 1.5-dimensional terrains.
\newblock In {\em Proceedings of the 13th Latin American Symposium on
  Theoretical Informatics}, pages 629--640, 2006.

\bibitem{conf/cccg/King08}
J.~King.
\newblock {VC}-dimension of visibility on terrains.
\newblock In {\em Proceedings of the 20th Canadian Conference on Computational
  Geometry}, 2008.

\bibitem{thesis/Kolen82}
A.~Kolen.
\newblock {\em Location problems on trees and in the rectilinear plane}.
\newblock PhD thesis, Matematisch Centrum, Amsterdam, 1982.

\bibitem{conf/stacs/Mestre08}
J.~Mestre.
\newblock Lagrangian relaxation and partial cover (extended abstract).
\newblock In {\em Proceedings of the 25th Annual Symposium on Theoretical
  Aspects of Computer Science}, pages 539--550, 2008.

\end{thebibliography}

\end{document}